\documentclass[11pt,a4paper,oneside,onecolumn]{article}

\usepackage{graphicx,amssymb,amsmath,theorem}
\usepackage{colortbl}
\usepackage{multicol}
\usepackage[ansinew]{inputenc}
\usepackage{lineno}
\usepackage{colortbl}

\usepackage[letterpaper]{geometry}
\newtheorem{theorem}{Theorem}[section]

\newtheorem{lemma}[theorem]{Lemma}

\def\QED{\ensuremath{{\square}}}
\def\markatright#1{\leavevmode\unskip\nobreak\quad\hspace*{\fill}{#1}}

\newenvironment{proof}
{\begin{trivlist}\item[\hskip\labelsep{\bf Proof.}]}
{\markatright{\QED}\end{trivlist}}

\def\eps{\varepsilon}

\title{Locating a service facility and a rapid transit line}

\author{
J. M. D\'{\i}az-B\'{a}\~{n}ez\thanks{Departamento de Matem\'{a}tica
Aplicada II, Universidad de Sevilla, Spain. Partially supported by
project MEC MTM2009-08652. {\tt
\{dbanez,iventura\}@us.es.}} \and
M. Korman\thanks{Université libre de Bruxelles (ULB). {\tt
mkormanc@ulb.ac.be.}} \and
P. P\'erez-Lantero\thanks{Escuela de Ingenieria Civil en Inform\'atica Departamento de Computaci\'{o}n, Universidad de Valparaíso, Chile. Partially supported by
project MEC MTM2009-08652 and grant FONDECYT 11110069. {\tt
pablo.perez@uv.cl}}\and
I. Ventura\footnotemark[1]}

\DeclareMathOperator{\sgn}{sgn} %

\newcommand{\Lx}{\Pi_x}
\newcommand{\Ly}{\Pi_y}
\newcommand{\Lu}{\Pi_{x+y}}

\newcommand{\x}[1]{x({#1})}
\newcommand{\y}[1]{y({#1})}

\graphicspath{{fig/}} 

\begin{document}
\maketitle

\linenumbers

\begin{abstract}
In this paper we study a facility location problem in the plane in which a single point
(facility) and a rapid transit line (highway) are simultaneously located in order to
minimize the total travel time of the clients to the facility, using the $L_1$ or
Manhattan metric. The rapid transit line  is represented by a line segment with fixed
length and arbitrary orientation. The highway is an alternative
transportation system that can be used by the clients to reduce their
travel time to the facility. This problem was introduced by Espejo
and Rodríguez-Chía in~\cite{espejo11}. They gave both a characterization of the optimal
solutions and an algorithm running in $O(n^3\log n)$ time, where $n$ represents the number
of clients. In this paper we show that Espejo and Rodríguez-Chía's algorithm does not
always work correctly.
At the same time, we provide a proper characterization of the solutions
and give an algorithm solving the problem in $O(n^3)$ time.
\end{abstract}

\textit{Keywords:} Geometric optimization; Facility location;
Transportation; Time distance.

\section{Introduction}

Suppose that we have a set of clients represented as a set of points in
the plane, and a service facility represented as a point to which
all clients have to move. Every client can reach the facility
directly or by using an alternative rapid transit line or highway,
represented by a straight line segment of
fixed length and arbitrary
orientation, in order to reduce the travel time. Whenever a client
moves directly to the facility, it moves at unit speed and the
distance traveled is the Manhattan or $L_1$ distance to the
facility. In the case where a client uses the highway, it travels
the $L_1$ distance at unit speed to one endpoint of the 
highway, traverses the entire highway with a speed greater than one,
and finally travels the $L_1$ distance from the other endpoint to the facility at unit speed. All clients traverse the
highway at the same speed. Given the set of points representing the clients,
the facility location problem consists in determining at the
same time the facility point and the highway in order to minimize
the \emph{total weighted travel time} from the clients to the facility. The
weighted travel time of a client is its travel time multiplied by
a weight representing the intensity of its demand. This problem was introduced
by Espejo and Rodríguez-Chía~\cite{espejo11}. We refer to~\cite{espejo11} and references
therein to review both the state of the art and
applications of this problem.

Geometric problems related to transportation networks have been recently 
considered in computational geometry.
Abellanas {\em et. al.} introduced the {\em time metric} model
in~\cite{abellanas03}: Given an underlying metric, the user can
travel at speed $v$ when moving along a highway $h$ or unit
speed elsewhere. The particular case in which the underlying metric
is the $L_1$ metric and all highways are axis-parallel segments of
the same speed, is called the {\em city metric}~\cite{aichholzer02}.
The optimal positioning of transportation systems that minimize the maximum travel time
among a set of points has been investigated in detail in recent papers~\cite{ahn07,
cardinal08,aloupis10}. Other more general models are studied in~\cite{korman08}. The
variant introduced by Espejo and
Rodríguez-Chía aims to minimize the sum of the travel times (transportation cost)  from
the demand points to the new facility service,
which has to be located  simultaneously with a highway. The highway is used by a demand
point whenever it saves time to reach the facility.

Notation to formulate the problem is as follows. Let $S$ be the set of $n$ client points;
$f$ the service facility point; $h$ the highway; $\ell$ the length of $h$; $t$ and
$t'$ the endpoints of $h$;  and
$v\geq 1$ the speed in which the points move along $h$. Let $w_p>0$ be the weight (or
demand) of a client point
$p$. Given a point $u$ of the plane, let $\x{u}$ and $\y{u}$ denote
the $x$- and $y-$coordinates of $u$ respectively. The distance or travel time
(see Figure~\ref{fig:distance}), between a point $p$ and the service
facility $f$ is given by the function
$$d_{t,t'}(p,f):=\min\left\{\|p-f\|_1,\|p-t\|_1+\frac{\ell}{v}+\|t'-f\|_1,
\|p-t'\|_1+\frac{\ell}{v}+\|t-f\|_1\right\}.$$

\begin{figure}[h]
    \centering
    \includegraphics[width=4.5cm]{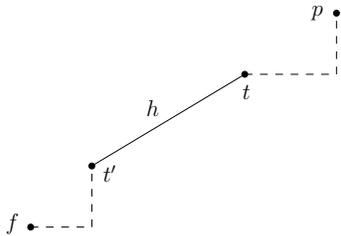}
    \caption{\small{The distance between a point $p$ and the facility $f$ using the highway.}}
    \label{fig:distance}
\end{figure}

Then the problem can be formulated as follows:

\begin{quote}
{\bf The Facility and Highway Location problem (FHL-problem)}: Given a set $S$ of $n$
points,  a weight $w_p>0$ associated with each point $p$ of $S$,   a fixed highway length
$\ell>0$, and a fixed speed $v\geq1$, locate a point (facility) $f$ and a line segment
(highway) $h$ of length $\ell$ with endpoints $t$ and $t'$ such that the function
$\sum_{p\in S}w_p \cdot d_{t,t'}(p,f)$ is minimized.
\end{quote}

Espejo and Rodríguez-Chía~\cite{espejo11} studied the FHL-problem and gave the
following characterization of the solutions. Consider the grid $G$
defined by the set of all axis-parallel lines passing through the
elements of $S$. They stated that there always exists an optimal
highway having one endpoint at a vertex of $G$. Based on this, they
proposed an $O(n^3\log n)$-time algorithm to solve the problem. In this paper we 
show that the
characterization given by Espejo and Rodríguez-Chía is not true in general, hence
their algorithm does not always give the optimal solution.

\paragraph{Addendum} An anonymous referee pointed out that the authors of~\cite{espejo11} published a {\em corrigendum} to their paper the 19th of January 2012, and that our result was not novel. In here we provide a chronological order of the events so that the reader can reach his/her own conclusions. The first version of this paper appeared the 5th of April 2011 on arXiv (and a preliminary version also appeared in the proceedings of the Spanish Meeting on Computational Geometry the 27th of June 2011). We contacted the authors of \cite{espejo11}, and provided them with a copy of our paper, including the counterexample. Naturally, they were interested in our research, and wanted to know where had they done a mistake. The 29th of October 2011, the authors of~\cite{espejo11} contacted us claiming that they had found the error in their paper. They provided us a write-up containing the corrected version of their proof, and suggested we combine our results. Given the difference in notation and the fact that this paper subsumes their result, we declined. From the conversation we can only deduce that the authors of~\cite{espejo11} submitted their corrigendum sometime in early November 2011.

As of now (16th March 2012), our paper is currently under supervision for journal publication, whereas the corrigendum has already appeared at COR. Although we would love if the submission, correction and publication process takes less than three months (as it appears to have happened with corrigendum at {\em Computers and Operations Research} journal), we understand that this is not possible in high-end journals. Regardless of our personal opinion of the actions of Espejo and Rodriguez-Chia, we believe that the date in which the result was found (and not published in a journal) is the relevant one. Thus, we claim that our paper is the first one to claim the error of~\cite{espejo11}. 

On a side note, we note that the corrigendum of Espejo and Rodriguez-Chia is also wrong, since they claim that our characterization is weaker. They specifically say that ``The description given by [this paper] means an infinite many number of candidates to be one of the endpoints of an optimal segment''. Although Lemma \ref{lemma:endpoint} does not explicitly say so, the algorithm of Section \ref{section:algorithm} only considers $O(n^3)$ cases (in particular a finite amount).

\paragraph{Paper Organization} In Section~\ref{section:properties} we first provide a proper
characterization of the solutions. 
After that we give a counterexample to the Espejo and
Rodríguez-Chía's characterization. We provide a set of
five points, all having weight equal to one, and prove that no
optimal highway has one endpoint in a vertex of $G$. In
Section~\ref{section:algorithm} we present an  improved algorithm running in
$O(n^3)$ time that correctly solves the FHL-problem. Finally, in
Section~\ref{section:conclusions}, we state our conclusions and proposal for
further research.

\section{Properties of an optimal solution}\label{section:properties}

A primary observation (also stated in~\cite{espejo11}) is that  the service facility can be
located at one of the endpoints of the rapid transit line. From now on, 
we assume $f=t'$ throughout the paper. This assumption simplifies the distance from
a point $p\in S$ to the facility to the following expression,

$$d_t(p,f)=\min\left\{\|p-f\|_1,\|p-t\|_1+\frac{\ell}{v}\right\}.$$

Using this observation, the expression of our objective function to minimize is
$\Phi(f,t)=\sum_{p\in S}w_p \cdot d_{t}(p,f)$. We call this value the {\em total
transportation cost} associated with $f$ and $t$ (or simply the {\em cost} of $f$ and
$t$). 

We say that a point $p$ uses the highway if $\|p-t\|_1+\frac{\ell}{v}<\|p-f\|_1$, and
that $p$ does not use it
(or goes directly to the facility) otherwise. Given $f$ and $t$, we call  \emph{travel bisector} of $f$ and $t$ (or \emph{ bisector} for short) as the set of points $z$ such that $\|z-f\|_1=\|z-t\|_1+\frac{\ell}{v}$,
see Figure~\ref{fig:bisector}. A geometrical description of such a bisector can be found
in~\cite{espejo11}, as the boundary of the so-called {\em captation region}.

\begin{figure}[h]
    \centering
    \includegraphics[width=14cm]{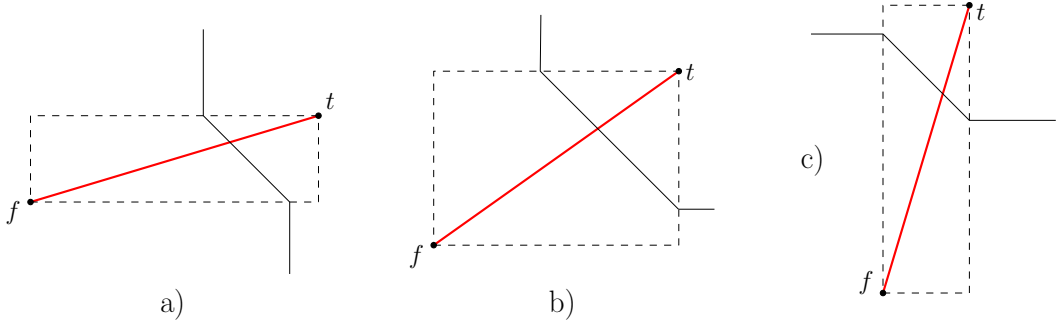} 
    \caption{\small{The bisector of $f$ and $t$.}}
    \label{fig:bisector}
\end{figure}

\begin{lemma}~\label{lemma:endpoint}
There exists an optimal solution to the FHL-problem satisfying one
of the next conditions:
\begin{itemize}
  \item[$(a)$] One of the endpoints of the highway is a vertex of $G$.
  \item[$(b)$] One endpoint of the highway is on a horizontal line
  of $G$, and the other endpoint is on a vertical line of $G$.
\end{itemize}
\end{lemma}

\begin{proof}
Let $f$ and $t$ be the endpoints of an optimal highway $h$ and assume neither of
conditions $(a)$ and $(b)$ is satisfied. Using local perturbation we will transform this
solution into one that satisfies one of these conditions. 
Assume neither $f$ nor $t$ is 
on any vertical line of $G$. 
Let $\delta_1>0$ (resp. $\delta_2>0$) be the smallest value such
that if we translate $h$ with vector $(-\delta_1,0)$ (resp. $(\delta_2,0)$)
then either one endpoint of $h$ touches a vertical line of $G$ or a demand point hits
the bisector of $f$ and $t$.
Given $\eps\in[-\delta_1,\delta_2]$, let $f_{\eps}$,
$t_{\eps}$, and $h_{\eps}$ be $f$, $t$, and $h$ translated with
vector $(\eps,0)$, respectively. It is easy to see that
$|d_{t_{\eps}}(p,f_{\eps})-d_t(p,f)|=|\eps|$ for all points $p$. Given
a real number $x$, let $\sgn(x)$ denote the sign of $x$. We partition $S$ into three sets
$S_1$, $S_2$ and $S_3$ as follows:
\begin{eqnarray*}
  S_1 & = & \{p\in S~|~\sgn(d_{t_{\eps}}(p,f_{\eps})-d_t(p,f))=\sgn(\eps),~~\forall\eps\in[-\delta_1,\delta_2]\setminus\{0\}\} \\
  S_2 & = & \{p\in S~|~\sgn(d_{t_{\eps}}(p,f_{\eps})-d_t(p,f))=-\sgn(\eps),~~\forall\eps\in[-\delta_1,\delta_2]\setminus\{0\}\} \\
  S_3 & = & \{p\in S~|~\sgn(d_{t_{\eps}}(p,f_{\eps})-d_t(p,f))=-1,~~\forall\eps\in[-\delta_1,\delta_2]\setminus\{0\}\}
\end{eqnarray*}

Observe that points of $S_3$ are in the bisector of $f$ and $t$; $S_1$ contains the
demand points that travel rightwards to reach $f$ directly or by using the highway, and $S_2$ contains the points that travel leftwards. 

Theoretically, one could consider the case in which a point belongs to set $S_4  =  \{p\in S~|~\sgn(d_{t_{\eps}}(p,f_{\eps})-d_t(p,f))=1,~~\forall\eps\in[-\delta_1,\delta_2]\setminus\{0\}\}$. Geometrically speaking, the points of this set are those that, when translating the highway in either directions, the distance between them and the entry point of the highway increases. This situation can only happen when the point is aligned with the entry point. That is, point $p\in S_4$ if and only if either $(i)$ $p$ uses the highway to reach the facility and it is vertically aligned with $t$, or $(ii)$ $p$ walks to the facility and it is vertically aligned with $f$. However, by definition of $\delta_1$ and $\delta_2$, no point of $S$ can belong to (or enter) $S_4$ during the whole translation.

By the linearity of the $L_1$ metric, whenever we translate the highway $\eps$ units to
the right (for some arbitrarily small $\eps$, $0<\eps\leq \delta_1$), the highway will be
$\eps$ units closer for points in $S_2\cup S_3$, but $\eps$ units further away for points
of $S_1$. Analogously, the distance to the facility
decreases for points in $S_1 \cup S_3$ and increases for points of $S_2$ when translating
$h$ leftwards. Let $N=\sum_{p\in S_1}w_p-\sum_{p\in S_2}w_p$ and
$k=\sum_{p\in S_3}w_p$. Thus, for any
vector $(\eps,0)$,
$\eps\in[-\delta_1,\delta_2]\setminus\{0\}$, the change of the
objective function when we translate the highway with vector $(\eps,0)$ is equal
to the following expression:
\begin{eqnarray*}
\Phi(f_{\eps},t_{\eps})-\Phi(f,t)&=&\sum_{p\in S}w_p\cdot
d_{t_{\eps}}(p,f_{\eps})-\sum_{p\in S}w_p\cdot d_{t}(p,f)\\
&=& \eps\sum_{p\in S_1}w_p-\eps\sum_{p\in S_2}w_p-|\eps|\sum_{p\in S_3}w_p\\
&=&N\eps-k|\eps|
\end{eqnarray*}
Since we initially assumed that the location of $h$ is optimal, we must have both $N=k=0$ (otherwise translating $h$ rightwards or leftwards would result in a decrease of the objective fuction). In particular, we can translate $h$ in either direction so that the cost of the objective function is unchanged.  

More importantly, observe that the value of $k$ must remain $0$ on the whole translation: if at some point it becomes positive we can find a translation from that point that reduces the cost of the objective function. In particular, the set $S_3$ must remain empty during the whole translation. Any point that changes from set $S_1$ to $S_2$ (or {\em vice versa}) must first enter $S_3$. Since the latter set remains empty during the whole translation, no point can change between sets $S_1,S_2$, or $S_3$ until either $f$ or $t$ is vertically aligned with a point of $S$.

We perform the
same operations on the
$y$ coordinates and obtain that one of the two endpoints is
on a horizontal line of $G$, hence satisfying one of the two conditions
of the Lemma.
\end{proof}

When the highway's length is equal to
zero, the FHL-problem is the weighted 1-median problem in metric $L_1$~\cite{durier1985},
and in this case the item (a) of Lemma~\ref{lemma:endpoint} holds.
Espejo and Rodríguez-Chía~\cite{espejo11} claimed that there always
exists an optimal solution of the FHL-problem that satisfies
Lemma~\ref{lemma:endpoint}~(a). Unfortunately, this claim is not true in general
and their algorithm may miss some highway locations; indeed, it may miss the optimal
location and thus fail. We provide here one counterexample and the
following result.

\begin{lemma}~\label{lemma:counter-example}
There exists a set of unweighted points in which no optimal solution to the
FHL-problem satisfies Lemma~\ref{lemma:endpoint}~(a).
\end{lemma}

\begin{proof}
Consider the problem instance with five points whose coordinates are $(-4,0)$, $(-3,-1)$, $(12,8)$,
$(13,5)$, and $(13,7)$, respectively (see Figure~\ref{fig:counter-example}). In the problem instance, we give unit weight to all points, and set the length $h$ of the highway as $\ell=\sqrt{180}\approx 13,5$. For simplicity in the calculations, we also set $v=\ell$, but any other large number works as well. The cost associated to the highway of endpoints $f=(12,6)$ and $t=(0,0)$ is $10+2\ell/v=12$. We claim that this location is better than any other solution with an endpoint at a vertex of $G$. 

\begin{figure}[h]
    \centering
    \includegraphics[width=12cm]{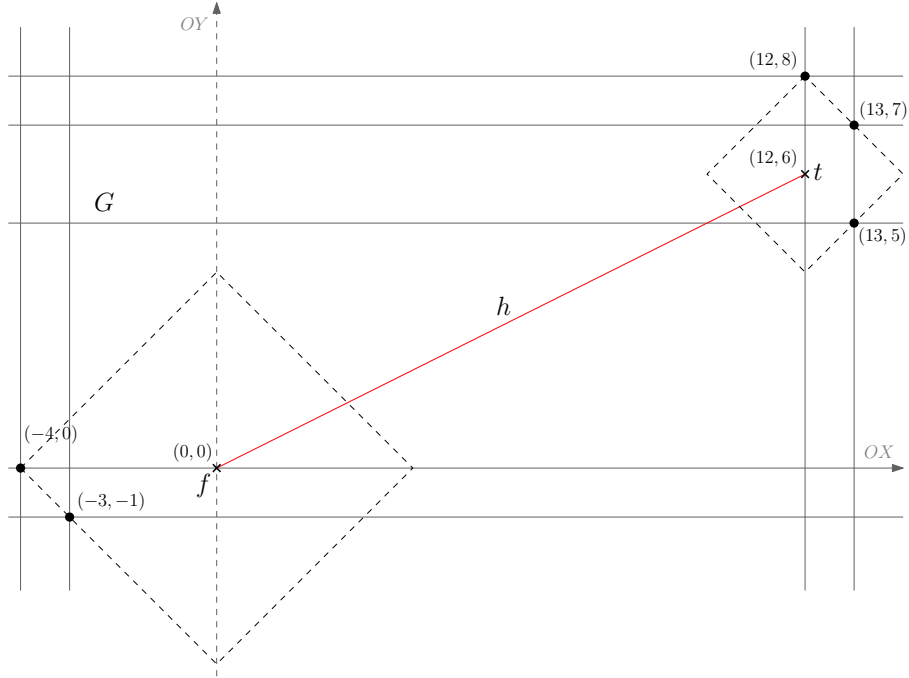} 
    \caption{\small{A counterexample to the algorithm of Espejo and Rodríguez-Chía.}}
    \label{fig:counter-example}
\end{figure}

If one endpoint of $h$ is a vertex of
$G$ in the line $x=-3$, then the other endpoint is located to the
left of the line $x=11$ because $-3+\ell<11$. In
that case we can translate $h$ rightwards with vector
$(\frac{1}{2},0)$ improving the objective function. The same holds
if one endpoint of $h$ is a vertex in the line $x=-4$. Similarly, if
one endpoint is a vertex in the line $x=13$, then we can translate
$h$ leftwards with vector $(-\frac{1}{2},0)$ and the objective
function decreases. 

Consider now locating one of the highway endpoints at coordinates $(12,0)$ or $(12,-1)$. Observe that the walking time (i.e., the traveling time when the highway is not used) from the points $(-4,0)$ and $(-3,-1)$ takes at least $15$ units of time, which is more than the cost associated with our solution. The same happens to the sum of the traveling times of the three other points. Hence, if $f$ is located at one of the two vertices, the five points must use highway (otherwise the travel time is higher than our solution). Analogously, if $t$ is located at grid points $(12,0)$ or $(12,-1)$, no point of $S$ will use the highway. In either case, the corresponding solution is at least as high as the sum of distances from all points of $S$ to the geometric median, which is higher than the cost associated with our solution.

Consider now the cases in which one of the endpoints has
coordinates $(12,y_0)$ for some $y_0\in\{5,7,8\}$. We start by showing
that, in any of the three cases, the optimal position of the other
endpoint of the highway (denoted by $e$) must lie on the line $y=0$.
Since the highway's length is equal to $\ell$, the possible positions of $e$ lie
both in circle $\sigma$ of radius $\ell$ centered at $(12,y_0)$ and to the left of line
$x=12$. 
Observe that the clients that
walk to $e$ are points $a=(-4,0)$ and $b=(-3,-1)$, located always
to the left of $e$. 
Hence, we are interested in minimizing the
expression $\|a-e\|_1+\|b-e\|_1$.
Let $a',b'\in\sigma$ denote respectively the closest points to $a$ and $b$
with the $L_1$ metric, which verify $y(a')=0$ and $y(b')=-1$. 
Observe that if $y(e)>0$ then $\|a-a'\|_1<\|a-e\|_1$ and $\|b-a'\|_1<\|b-e\|_1$ implying
$$\|a-a'\|_1+\|b-a'\|_1<\|a-e\|_1+\|b-e\|_1$$
(see Figure~\ref{fig:counter-example2} a)). Similarly, if $y(e)<-1$, then
$$\|a-b'\|_1+\|b-b'\|_1<\|a-e\|_1+\|b-e\|_1.$$
Therefore, $e$ must satisfy $-1\leq y(e)\leq 0$ (see Figure~\ref{fig:counter-example2}
b)). In this case we have
\begin{eqnarray*}
\|a-e\|_1+\|b-e\|_1&=&x(e)-x(a)+y(a)-y(e)+x(e)-x(b)+y(e)-y(b)\\
  &=& 2x(e) + 8
\end{eqnarray*}
Then $\|a-e\|_1+\|b-e\|_1$ is minimized when $x(e)$ is minimum, and it
happens when $y(e)=0$.
\begin{figure}[h]
    \centering
    \includegraphics[width=10cm]{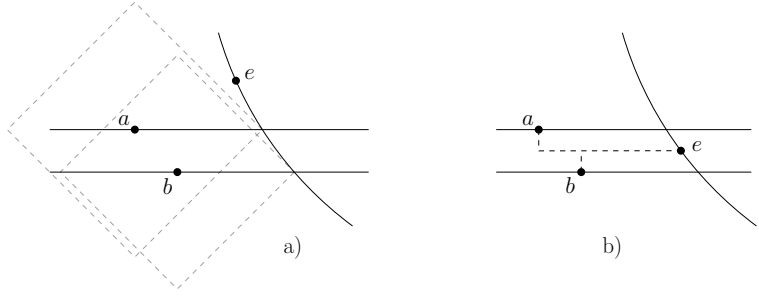} 
    \caption{\small{$a=(-4,0)$ and $b=(-3,-1)$. 
    When one endpoint of the highway has coordinates $(12,8)$, $(12,7)$, or $(12,5)$, the
    optimal position of the other endpoint $e$ is on the line $y=0$.}}
    \label{fig:counter-example2}
\end{figure}

If $y_0=8$, then $h$ can be translated downwards with vector
$(0,-\frac{1}{2})$ and the value of the objective function decreases. 
Thus point $(12,8)$ is discarded. It remains
to show that there is a solution better than the one having an
endpoint at either $(12,7)$ or $(12,5)$, and the other endpoint on
the line $y=0$. Observe that if $f$ and $t$ belong to the lines $y=0$
and $x=12$, respectively, then by exchanging $f$ and $t$ the value
of the objective function reduces in $\ell/v$. Then consider the case
where $y(t)=0$ and $x(f)=12$.  

Let $t=(0,0)$ and $f=(12,6)$.
Given a value $\eps$, let $t_{\eps}$ be the point with coordinates
$(\eps,0)$ and $f_{\eps}$ be the point in the line $x=12$ such that
$\y{f_{\eps}}>0$ and the Euclidean distance between $f_{\eps}$ and
$t_{\eps}$ is equal to $\ell$ (see
Figure~\ref{fig:counter-example3}). Let $[-\delta_1,\delta_2]$,
$\delta_1,\delta_2>0$, be the maximal-length interval such that $5
\leq \y{f_{\eps}}\leq 7$ for all $\eps\in[-\delta_1,\delta_2]$.
Note $\delta_1=\sqrt{155}-12<1$ and $\delta_2=12-\sqrt{131}<1$. Then $|\eps|<1$.
\begin{figure}[h]
    \centering
    \includegraphics[width=8.5cm]{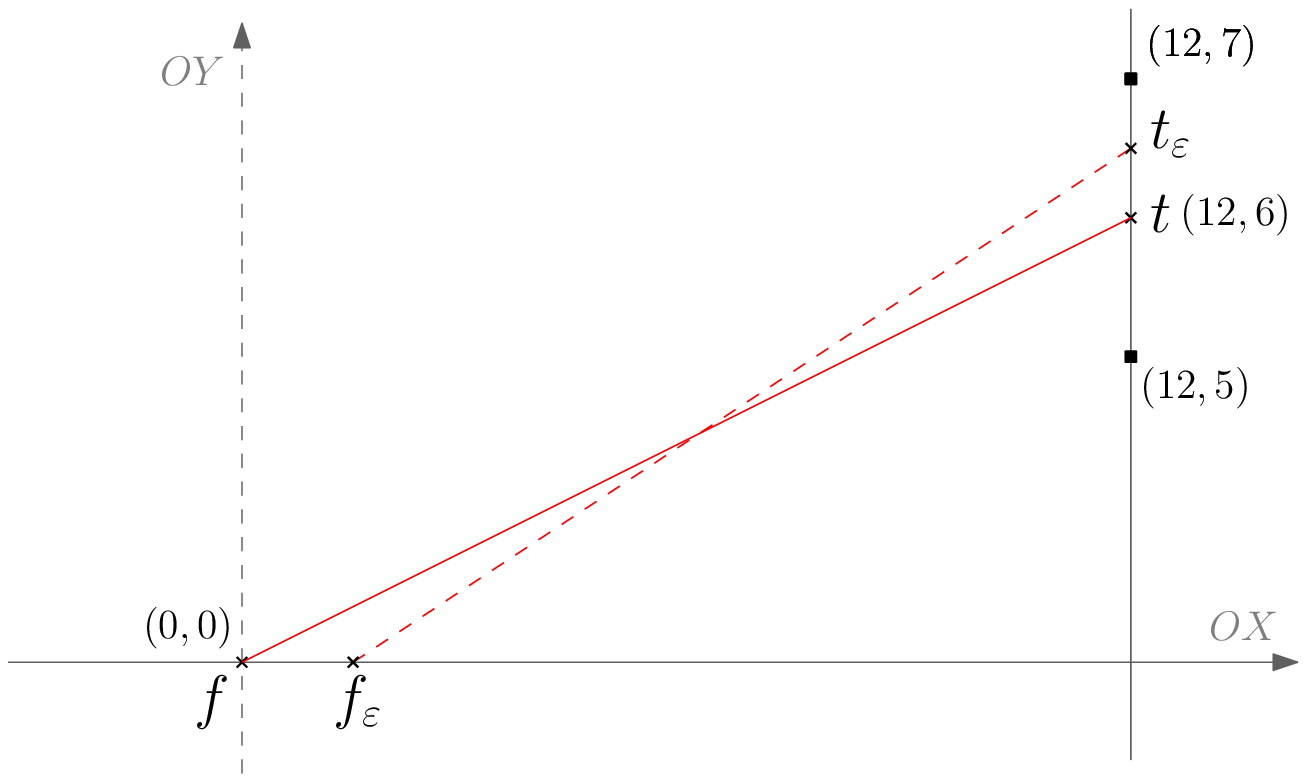}
    \caption{\small{Definitions of $f_{\eps}$ and $t_{\eps}$.}}
    \label{fig:counter-example3}
\end{figure}

The variation of the objective function's value when $f$ and $t$ are moved
to $f_{\eps}$ and $t_{\eps}$, respectively, is equal to
\begin{eqnarray*}
 g(\eps)&:=&\Phi(f_{\eps},t_{\eps})-\Phi(f,t)\\
&=&2\left(\x{t_{\eps}}-\x{t}\right)-\left(\y{f_{\eps}}-\y{f}\right)\\
&=&2\eps-\left(\sqrt{36+24\eps-\eps^2}-6\right).
\end{eqnarray*}
In the following we will show that $\sqrt{36+24\eps-\eps^2}<6+2\eps$, for all $\eps\in[-\delta_1,\delta_2]\setminus\{0\}$. In particular, we will have $g(\eps)>0$ (except when $\eps=0$), implying that our highway location is optimal. First observe that 
$4\eps^2+24\eps+36=(2\eps+6)^2>36+24\eps-\eps^2$. 
Since $|\eps|<1$ then $2\eps+6>0$ and $36+24\eps-\eps^2>0$, which implies
$2\eps+6>\sqrt{36+24\eps-\eps^2}$. Thus $g(\eps)>0$ and 
the highway with endpoints $f$ and $t$ gives a better solution than that
having an endpoint at $(12,7)$ or $(12,5)$. This completes the proof.
\end{proof}

In the next section we provide a correct algorithm
that solves the problem in $O(n^3)$ time.
We assume general position, that is, there are no two points on a
same line having slope in the set $\{-1,0,1,\infty\}$.

\section{The algorithm}\label{section:algorithm}

Lemma~\ref{lemma:endpoint} can be used to find an optimal solution
to the FHL-problem. Although the method is quite similar for both cases in
Lemma~\ref{lemma:endpoint}, we address the two cases independently for the sake of
clarity.
By Vertex-FHL-problem we will denote the FHL-problem for the cases in which
Lemma~\ref{lemma:endpoint} a) holds, and by
Edge-FHL-problem the FHL-problem for the cases in which
Lemma~\ref{lemma:endpoint} b) holds. In the next subsections we give an
$O(n^3)$-time algorithm for each variant of the problem. In both of
them we assume w.l.o.g. that highway's length $\ell$ is equal to
one.

In the following $\theta$ will denote the positive angle of the highway with
respect to the positive direction of the $x$-axis. For the sake of clarity, we will assume
that $\theta\in[0,\frac{\pi}{4}]$. When $\theta$ belongs to
the interval $[k\frac{\pi}{4},(k+1)\frac{\pi}{4}]$, $k=1,\dots,7$,  both  the
Vertex- and  Edge-FHL-problem
can be solved in a similar way.

Given a point $u$ and an angle $\theta$, let $u(\theta)$ be the
point with coordinates $(\x{u} + \cos\theta, \y{u} + \sin\theta)$.
There exists an angle
$\phi\in[0,\frac{\pi}{4}]$ such that the bisector of the endpoints $f$ and
$t=f(\theta)$ has the shape in Figure~\ref{fig:bisector} a) for all
$\theta\in[0,\phi)$, and has the shape in Figure~\ref{fig:bisector}
b) for all $\theta\in(\phi,\frac{\pi}{4}]$. Such an angle $\phi$ verifies
$\cos(\phi)-\sin(\phi)=\frac{1}{v}$. Furthermore,
$\phi=\frac{1}{2}\arcsin(1-\frac{1}{v^2})$ and
$\phi\neq\frac{\pi}{4}$ unless $v$ is infinite. Refer to~\cite{espejo11}
for a detailed description of this situation.

Let $\Lx$, $\Ly$, and $\Lu$ denote the point set $S$  sorted according to the $x$-, $y$-,
and $(x+y)$-order, respectively.

\subsection{Solving the Vertex-FHL-problem}

For each vertex $u$ of $G$ we can solve the problem subject to $f=u$ or $t =u$. We show
how to obtain a solution if $f=u$. The case where $t =u$ can be solved analogously.

Suppose w.l.o.g. that the vertex $f=u$ is the origin of the coordinate system and the
highway angle 
is $\theta$, for $\theta\in[0,\frac{\pi}{4}]$. 
Then $t=u(\theta)=(\cos\theta,\sin\theta)$.
and the distance $d_t(p,f)$ between a point $p\in S$ and the facility $u$ has the
expression $c_1+c_2\cos\theta+c_3\sin\theta$, where $c_1,c_2,c_3$
are constants satisfying
$c_2,c_3\in\{-1,0,1\}$.
When $\theta$ goes from $0$ to
$\frac{\pi}{4}$, this expression changes at the values of $\theta$
such that:

\begin{itemize}
    \item The point $p$ switches from using the highway to going directly to
    the facility (or vice versa). We call these changes \emph{bisector
    events}. A bisector event occurs when the bisector between the highway's endpoints
    $u$ and $u(\theta)$, contains $p$. At most two bisector events are obtained for each point $p$.

    \item The highway endpoint $u(\theta)$ crosses the vertical or horizontal
    line passing through $p$. We call this event \emph{grid event}. Again, each
    point of $S$ generates at most two grid events.

    \item $\theta=\phi$. We call it the $\phi$\emph{-event}.
\end{itemize}

We refer the interested reader to~\cite{espejo11} for a detailed description of the above events\footnote{Although their events are very similar to the ones we described, the authors of~\cite{espejo11} refer to them as {\em projection} and {\em limit points}. We prefer to use the term ``event'', since ``point'' is reserved for the elements of $S$}. The cost of their algorithm is dominated by the time spent sorting the order in which events take place. In order to avoid this time, we use the following result: 

\begin{lemma}\label{lemma:events}
After an $O(n\log n)$-time preprocessing, the angular order of all
the events associated with a given vertex of $G$ can be obtained in
linear time.
\end{lemma}

\begin{proof}
The preprocessing consists in computing $\Lx$, $\Ly$, and $\Lu$,
which can be done in $O(n\log n)$ time. Now, let $u$ be a vertex of
$G$. It is straightforward to see that they are $O(n)$ grid events
and that we can obtain their angular order in linear time by using
both $\Lx$ and $\Ly$. Let us show how to obtain  the
bisector events in $O(n)$ time.

The bisector of $u$ and $u(\theta)$ consists of two axis-aligned
half-lines and a line segment with slope -1 connecting their endpoints
(see Figure~\ref{fig:bisector} and~\cite{espejo11} for further details). Given a point
$p$, when $\theta$ goes from $0$ to $\pi/4$ the bisector between $u$
and $u(\theta)$ passes through $p$ at most twice, that is, when $p$
belongs to one of the half-lines of the bisector and when $p$
belongs to the line segment. If $p$ belongs to the line segment of the
bisector  then the event is denoted by $\alpha_p$ (see
Figure~\ref{fig:bisector-events} b)). If $p$ belongs to the leftmost
half-line of the bisector, which is always vertical, we denote that
event by $\beta_p$ (see Figure~\ref{fig:bisector-events} a)).
Otherwise, if $p$ belongs to the rightmost half-line which can be
either vertical or horizontal we denote that event by $\gamma_p$
(see Figure~\ref{fig:bisector-events} c) and d)). Observe that if
the rightmost half-line is vertical then $\gamma_p<\phi$, otherwise
$\gamma_p>\phi$. Refer to~\cite{espejo11} for a characterization to 
identify whether a point $p\in S$ generates a bisector event for some angle $\theta$.

\begin{figure}[h]
    \centering
    \includegraphics[width=12cm]{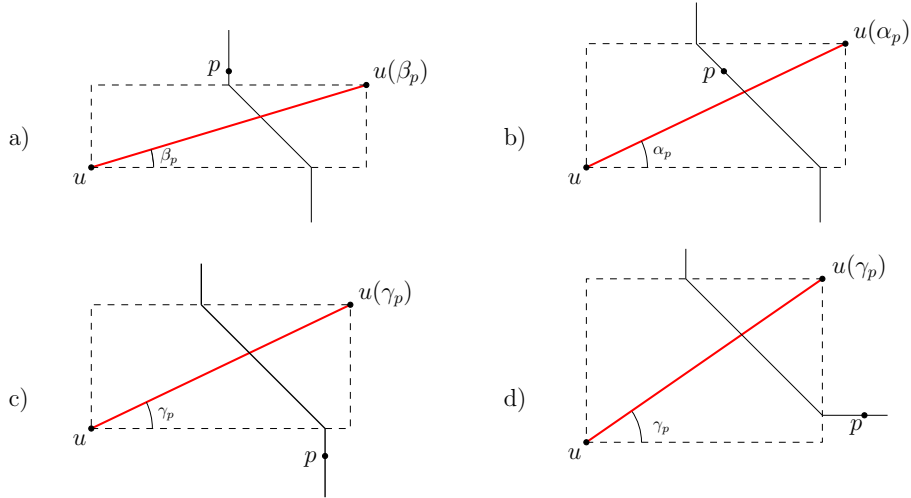}
    \caption{\small{The bisector events of $p$ when $\theta\in[0,\frac{\pi}{4}]$. a) $p$ belongs
    to the leftmost half-line of the bisector of $u$ and $u(\theta)$. b) $p$ belongs to the
    segment.
    c,d) $p$ belongs to the rightmost half-line of the bisector.}}
    \label{fig:bisector-events}
\end{figure}

Let $\Pi_1$ be the subsequence of $\Lu$ containing all elements $p$
such that $\alpha_p \in [0,\frac{\pi}{4}]$, $\Pi_2$ be the
subsequence of $\Lx$ containing all elements $p$ such that $\beta_p \in
[0,\frac{\pi}{4}]$, and $\Pi_3$ be the subsequence of
$\Lx$ that contains all elements $p$ such that $\y{p}<\y{u}$ and
$\gamma_p \in [0,\phi]$, concatenated with the
subsequence of $\Ly$ that contains all elements $p$ such that
$\x{p}>\x{u}$ and $\gamma_p \in [\phi,\frac{\pi}{4}]$. Given a
point $p\in S$, the corresponding events of $p$ in $[0,\frac{\pi}{4}]$
can be found in constant time, thus $\Pi_1$,
$\Pi_2$, and $\Pi_3$ can be built in linear time.

The following statements are true for any point $p\in S$:
\begin{itemize}
  \item[$(a)$] $\x{p}+\y{p}=\frac{1}{2}(\cos\alpha_p+\sin\alpha_p+\frac{1}{v})$ for all points $p$ in $\Pi_1$.
  \item[$(b)$] $\x{p}=\frac{1}{2}(\cos\beta_p-\sin\beta_p+\frac{1}{v})$ for all points $p$ in $\Pi_2$.
  \item[$(c)$] $\x{p}=\frac{1}{2}(\cos\gamma_p+\sin\gamma_p+\frac{1}{v})$ for all points $p$ in $\Pi_3$ such that $\gamma_p<\phi$.
  \item[$(d)$] $\y{p}=\frac{1}{2}(-\cos\gamma_p+\sin\gamma_p+\frac{1}{v})$ for all points $p$ in $\Pi_3$ such that $\gamma_p>\phi$.
\end{itemize}

Let $\Gamma_1$ (resp. $\Gamma_2$, $\Gamma_3$) be the sequence
obtained by replacing each element $p$ in $\Pi_1$ (resp. $\Pi_2$,
$\Pi_3$) by $\alpha_p$ (resp. $\beta_p$, $\gamma_p$). Therefore,
from statements $(a)-(d)$ and the monotonicity of the functions
$\cos\theta+\sin\theta$, $\cos\theta-\sin\theta$, and
$-\cos\theta+\sin\theta$ in the interval $[0,\frac{\pi}{4}]$, we
obtain that $\Gamma_1$, $\Gamma_2$, and $\Gamma_3$ are sorted
sequences. Using a standard method for merging sorted lists, we can
merge in linear time $\Gamma_1$, $\Gamma_2$, $\Gamma_3$, the grid
events, and the $\phi$-event. Therefore, the angular order of all
events associated with a vertex $u$ can be obtained in $O(n)$ time and the result follows.
\end{proof}

\begin{theorem}\label{theorem:vertex}
The Vertex-FHL-problem can be solved in $O(n^3)$ time.
\end{theorem}

\begin{proof}
Let $u$ be a vertex of $G$. Using Lemma~\ref{lemma:events}, we
obtain in linear time the angular order of the $O(n)$ events
associated with $u$. The events induce a partition of
$[0,\frac{\pi}{4}]$ into maximal intervals. For each of those
intervals, the objective function takes the form $g(\theta):=
\Phi(f,t)=\Phi(u,u(\theta))=b_1+b_2\cos\theta+b_3\sin\theta$,
where $b_1,b_2,b_3$ are constants.
This problem is of constant size in each subinterval and the minimum of
$g(\theta)$ can be found in $O(1)$ time. Furthermore, the expression of
$g(\theta)$ can be updated in constant time when $\theta$ crosses an
event point distinct of $\phi$ when going from $0$ to
$\frac{\pi}{4}$. In the case where $\theta$ crosses $\phi$,
$g(\theta)$ can be updated in at most $O(n)$ time. Then the problem
subject to $f=u$ can be solved in linear time. The case in which $t=u$ can be addressed in a similar way.
It gives an overall
$O(n^3)$ time complexity because $G$ has $O(n^2)$ vertices. \end{proof}

\subsection{Solving the Edge-FHL-problem}

We now consider the case in which the optimal solution satisfies condition b) of  Lemma~\ref{lemma:endpoint}.
Namely, we consider a horizontal line $e_h$ of $G$ and each vertical line
$e_v$ of $G$. For every pair of such lines, we consider eight different sub-cases, depending on whether $h$ is located above/below $e_h$, rightwards/leftwards of $e_v$, and $f\in e_h$ and $t\in e_v$ (or {\em vice versa}). For a fixed sub-case, we parametrize the location of the highway by the angle $\theta$ that the highway forms with $e_h$. As in the Vertex-FHL case, we assume that $f\in e_h$, $t\in e_v$, and $\theta\in[0,\frac{\pi}{4}]$. 

We implicitly redefine the coordinate system so that $e_h$ and $e_v$  intersect at the origin $o$. Let $\theta\in[0,\frac{\pi}{4}]$ be the positive angle of the highway with respect to
the positive direction of the $x$-axis and $f=x_{\theta}$, $t=y_{\theta}$ be the highway endpoints, see Figure~\ref{fig:edge-FHL}. 

First notice that, since we are again doing a continuous translation of $h$, the events that affect the value of the objective function are exactly the same as those that happen in the Vertex-FHL-problem: bisector-, grid- and $\phi$- events. We start by showing that the equivalent of Lemma \ref{lemma:events} also holds: 

\begin{figure}[h]
    \centering
    \includegraphics[width=8cm]{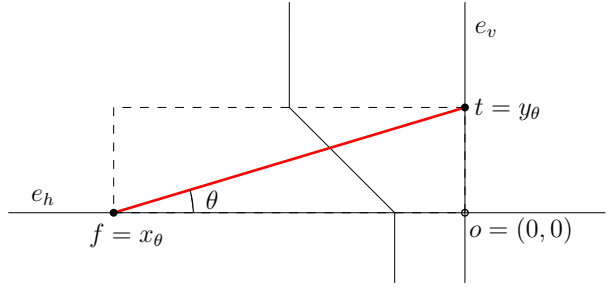}
    \caption{\small{Solving the Edge-FHL-problem.}}
    \label{fig:edge-FHL}
\end{figure}

\begin{lemma}\label{lemma:events2}
After an $O(n\log n)$-time preprocessing, the angular order of all
the events associated with a pair  of perpendicular lines of $G$
 can be obtained in linear time.
\end{lemma}

\begin{proof}

We can follow the arguments of Lemma~\ref{lemma:events}. 
Firstly, we note that there are $O(n)$ grid
events and their angular order can be obtained in linear time by
using both $\Lx$ and $\Ly$.

Given a point $p\in S$, let the events $\alpha_p$, $\beta_p$, and
$\gamma_p$ be defined as in the Vertex-FHL case.
Refer to Figure~\ref{fig:bisector-events}. Let $\Pi_1$ be the
subsequence of $\Lu$ containing all elements $p$ such that
$\alpha_p \in [0,\frac{\pi}{4}]$, $\Pi_2$ be the subsequence
of $\Lx$ containing all elements $p$ such that $\beta_p  \in
[0,\frac{\pi}{4}]$, and $\Pi_3$ be the subsequence of $\Lx$ that
contains all elements $p$ such that $\y{p}<\y{o}$ and $\gamma_p \in
[0,\phi]$, concatenated with the subsequence of
$\Ly$ that contains all elements $p$ such that $\x{p}>\x{o}$ and
$\gamma_p \in [\phi,\frac{\pi}{4}]$. Note that
$\Pi_1$, $\Pi_2$, and $\Pi_3$ can be built in linear time.

Given a point $p\in S$, the following statements are true:
\begin{itemize}
  \item[$(a)$] $\x{p}+\y{p}=\frac{1}{2}(-\cos\alpha_p+\sin\alpha_p+\frac{1}{v})$ for all points $p$ in $\Pi_1$.
  \item[$(b)$] $\x{p}=\frac{1}{2}(-\cos\beta_p-\sin\beta_p+\frac{1}{v})$ for all points $p$ in $\Pi_2$.
  \item[$(c)$] $\x{p}=\frac{1}{2}(-\cos\gamma_p+\sin\gamma_p+\frac{1}{v})$ for all points $p$ in $\Pi_3$ such that $\gamma_p<\phi$.
  \item[$(d)$] $\y{p}=\frac{1}{2}(-\cos\gamma_p+\sin\gamma_p+\frac{1}{v})$ for all points $p$ in $\Pi_3$ such that $\gamma_p>\phi$.
\end{itemize}

Let $\Gamma_1$ (resp. $\Gamma_2$, $\Gamma_3$) be the sequence
obtained by replacing each element $p$ in $\Pi_1$ (resp. $\Pi_2$,
$\Pi_3$) by $\alpha_p$ (resp. $\beta_p$, $\gamma_b$). Therefore,
by using similar arguments to those used in Lemma~\ref{lemma:events} the angular order of
all events can be obtained in $O(n)$ time,
once the lists $\Pi_x$, $\Pi_y$
and $\Pi_{x+y}$ have been precomputed.
\end{proof}

Consider now a small interval $[\theta_1,\theta_2]$ in which no event occurs. Observe that, after the coordinate system redefinition, we have $f=x_{\theta}=(-\cos\theta,0)$, and $t=y_{\theta}=(0,\ell\sin \theta)$. Let $p\in S$ be a point that uses the highway to reach the facility; since only the $y$-coordinate of $t$ changes, its distance to $f$ can be expressed as $c_1\pm\sin\theta$ for some $c_1>0$. Analogously, if $p$ walks to $f$, its distance is of the form $c_1\pm\cos\theta$ for some $c_1>0$. That is, the distance between a point of $S$ and $f$ in any interval is of the form $c_1+c_2\sin\theta+c_3\cos\theta$ for some constants $c_1>0$ and $c_2,c_3\in\{-1,0,1\}$.

\begin{theorem}
The Edge-FHL-problem can be solved in $O(n^3)$ time.
\end{theorem}

\begin{proof}
We can use a method similar to the one used in the Vertex-FHL-problem.
Let $e_h$ be a horizontal line of $G$ and $e_v$ be a vertical line
of $G$. 

Using Lemma~\ref{lemma:events2}, we obtain in linear time
the angular order of the $O(n)$ events associated with $e_h$ and
$e_v$. The events induce a partition of $[0,\frac{\pi}{4}]$ into
maximal intervals. For each of those intervals the objective
function has the form
$g(\theta):=\Phi(f,t)=\Phi(x_{\theta},y_{\theta}
)=b_1+b_2\cos\theta+b_3\sin\theta$, where $b_1>0$, and $b_2,b_3\in\mathbb{Z}$ are constants. 
This problem has constant
size, hence the minimum of $g(\theta)$ can be found in $O(1)$ time.
Furthermore, the expression of $g(\theta)$ can be updated in
constant time when $\theta$ crosses an event point distinct of
$\phi$ when it goes from $0$ to $\frac{\pi}{4}$. In the case where
$\theta$ crosses $\phi$, $g(\theta)$ can be updated in at most
$O(n)$ time. Then the problem subject to $f\in e_h$ and $t\in e_v$
can be solved in linear time. It gives an overall $O(n^3)$ time
complexity because $G$ has $O(n^2)$ pairs consisting of a horizontal and a
vertical line.
\end{proof}

\section{Experimental results}\label{sec_exper}
Similar to \cite{espejo11}, we explore examples of solutions to the FHL-problem for different values
of the length of the line segment. The problem instance is given by the unweighted points
with coordinates $(-4,0)$, $(-3,-1)$, $(12,8)$,
$(13,5)$, and $(13,7)$ as in Lemma~\ref{lemma:counter-example} and we consider locating a
highway for different values of length and speed. Given a fixed value of speed, say $v=2$,
Figure~\ref{fig_examples} shows the location of the optimal highways for some values of
$\ell$. Note that the case $\ell=0$ is the Fermat-Weber problem for the $L_1$-metric. The
highway's length and the associated total transportation cost for each of these solutions
can be seen in Table~\ref{tab_examples}. The optimal solution for each of the cases (and
its associated cost) has been obtained with the help of a computer.

Observe that, for some values of $\ell$, the optimal solution satisfies condition $(a)$
of Lemma~\ref{lemma:endpoint}, but in other situations condition $(b)$ is satisfied
instead (see Figure~\ref{fig_examples} d), where the highway's length has been set to
$13.41$). Experimentally we observed that increasing the highway's length decreases the
total transportation cost until $\ell=\sqrt{305}$, in which a total cost of $5+2\ell/v$ is
obtained (see Figure~\ref{fig_examples} e)). Afterwards the cost gradually increases until
we locate a highway so long that no point of $S$ uses it to reach $f$. We also note that
for this demand point set the highway's speed has a small impact on the optimal solution.
Indeed,  increasing the highway's speed changes the total cost but  the location of the
highway in the above instance is unaffected by the highway's speed (provided that $v>1$).
The fourth column in Table~\ref{tab_examples} gives the small variation of the total cost
with respect to the speed. This suggests the following open problem: given an
instance of the FHL-problem, can we efficiently compute the highway's
length that minimizes the total transportation cost?

\begin{figure}
    \centering
    \includegraphics[scale=0.5]{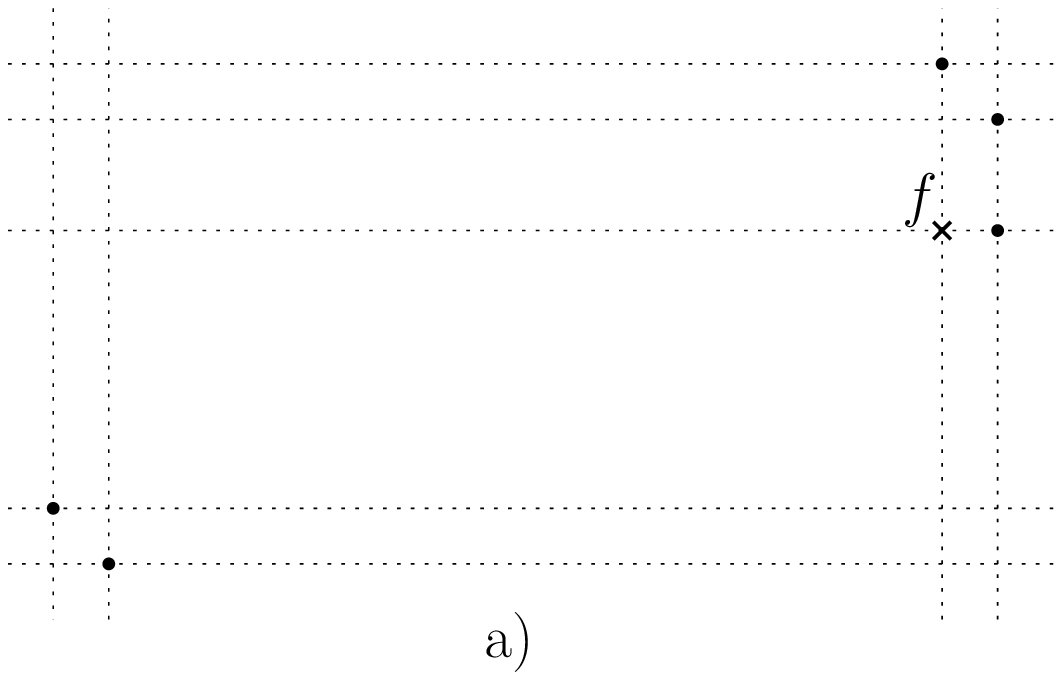}
    \includegraphics[scale=0.5]{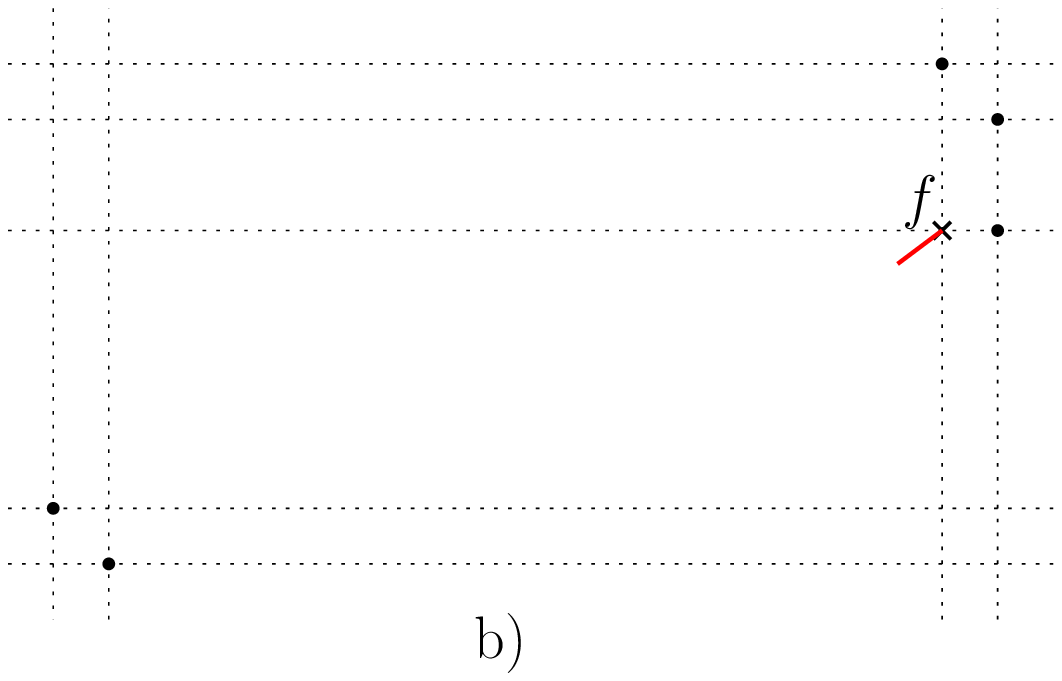}
    \includegraphics[scale=0.5]{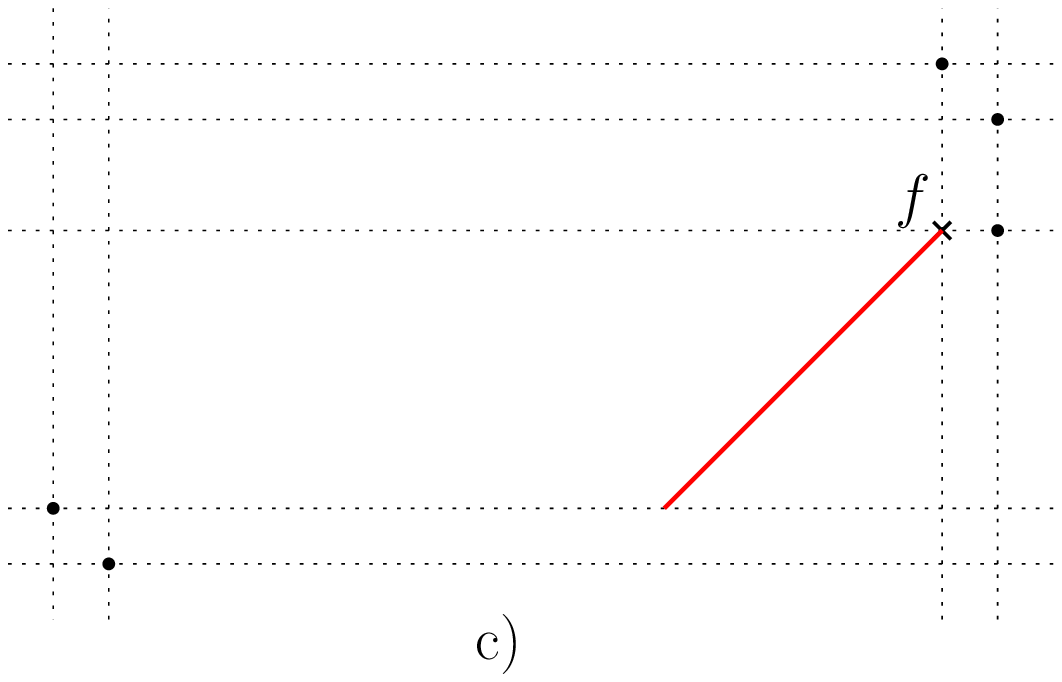}
    \includegraphics[scale=0.5]{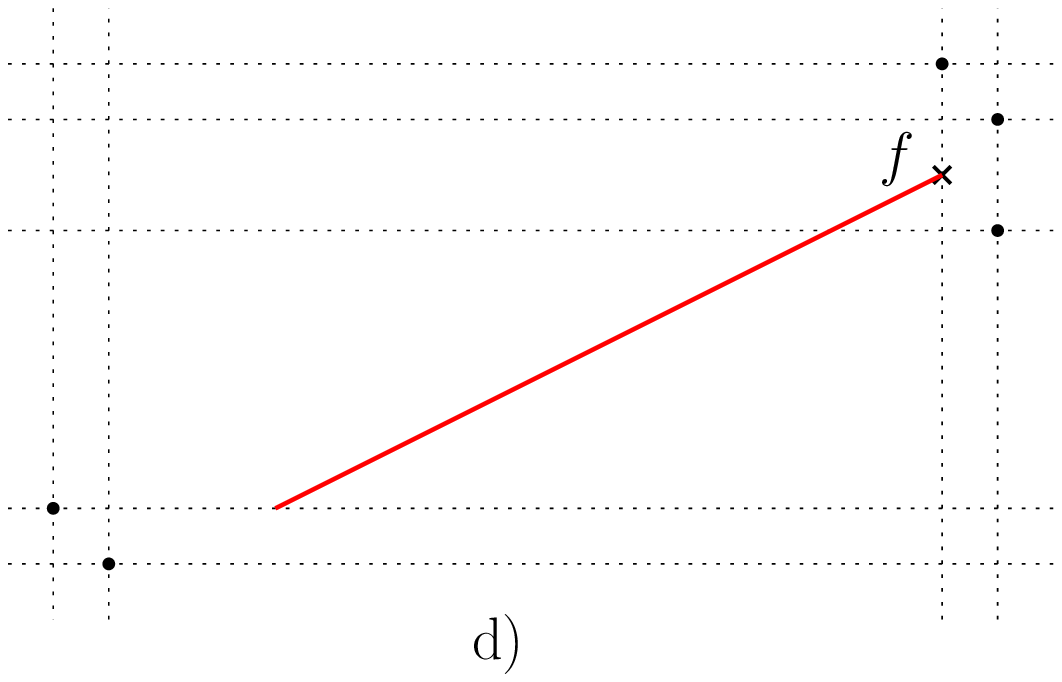}
    \includegraphics[scale=0.5]{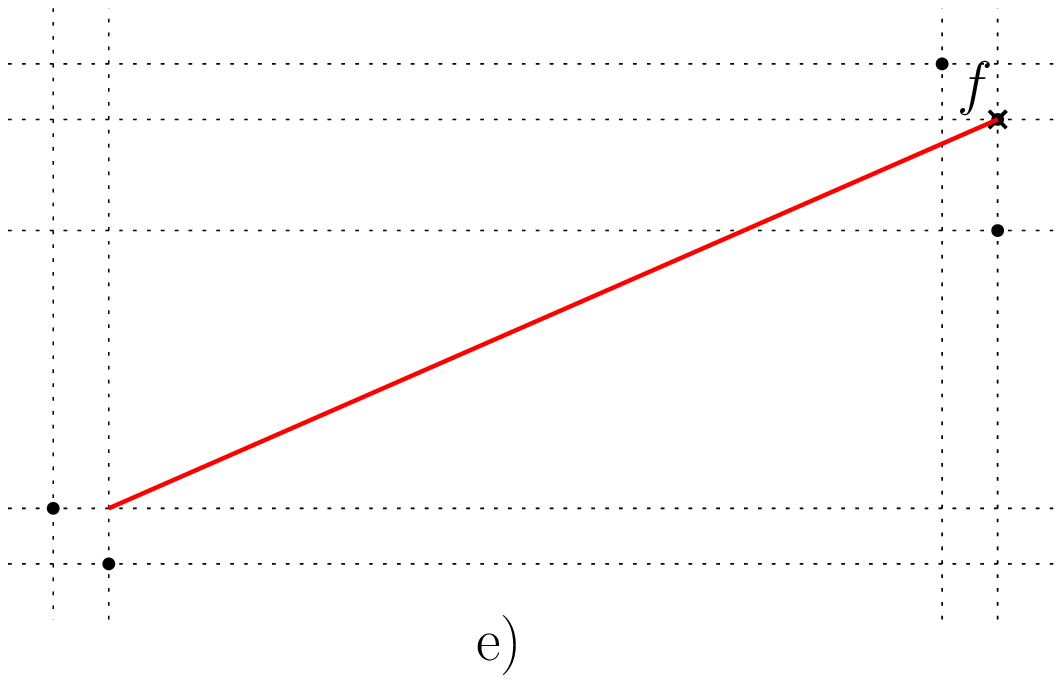}    
    \caption{\small{Solution of the same instance of the FHL-problem for different values
of $\ell$. The optimal highway is depicted in red (and the endpoint containing $f$ as a
cross). The exact highway's length and the associated total transportation cost can be
seen in Table \ref{tab_examples}.}}
    \label{fig_examples}
\end{figure}

\begin{table}
\small
\label{tab_examples}
\centering
\begin{tabular}{|c|c|c|c|c|}
\hline
\rowcolor[gray]{0.95}\hspace*{0.1cm} Figure~\ref{fig_examples} & $\ell$ & $v$ & Cost & Ratio\\ \hline \hline
a) & 0 & - & 49 & 1 \\ \hline
b) & 1 & 2 & 46 &0.93 \\
   &  & 4 & 45.5 & 0.92 \\
   & & $10^6$ & 45 & 0.91 \\ \hline
c) & 7.07 & 2 & 34.07 & 0.7 \\
   &  & 4 & 30.54 & 0.62 \\
   &  & $10^6$ & 27 & 0.55 \\ \hline
d) & 13.41 & 2 & 27.41 &0.56 \\
   &  & 4 & 20.71 & 0.42 \\
   &  & $10^6$ & 14 & 0.29 \\ \hline
e) & 16.55 & 2 & 9 & 0.18 \\
   & & 4 & 8.5 & 0.17 \\
   & & $10^6$ & 8 & 0.16 \\ \hline
\end{tabular}
\caption{\small{Total transportation cost as a function of the highway's length and speed.
The last column shows how much does the highway improve the total transportation cost
(compared to the case in which only a facility is located)}}
\end{table}

\section{Concluding remarks}\label{section:conclusions}


As further research, it would be worth studying the same problem in other
metrics or using different optimization criteria. Another interesting variant would be to
consider the problem when the length of the highway is not given in advance and it is a
variable in the problem. Additionally, we could consider a similar distance model  in
which the clients can enter and exit  the highway at any point (called \emph{freeway}
in~\cite{bae09}). 

Motivated from the experimental results of Section~\ref{sec_exper}, we can deduce that
the highway's length has a strong impact on the optimal solution. As one would expect,
when the highway's length is small, the total cost barely changes. We obtain a similar
effect when the highway to locate is very long, since traveling to the opposite endpoint
takes more time than walking directly to the facility. Hence, it would be interesting to
consider a variation of the problem in which we can also adjust the highway's length.
Specially, one would like to find a balance between the cost of constructing a longer
highway and the improvement in the total transportation cost. 

\small

\bibliographystyle{plain}
\bibliography{carreteras2}{}

\end{document}